\documentclass[prd,preprint,nofootinbib]{revtex4} 
\usepackage{graphicx,amsmath,amsfonts,amssymb,revsymb,dcolumn,amsthm,epsf} 
\usepackage[utf8]{inputenc}

\newcommand{\be}{\begin{equation}}
\newcommand{\ee}{\end{equation}}
\newcommand{\bea}{\begin{eqnarray}}
\newcommand{\eea}{\end{eqnarray}}

\newtheorem{Def}{Definition}
\newtheorem*{theorem*}{Proposition}
\begin{document}

\title{Complete inequivalence of nonholonomic and vakonomic mechanics: rolling coin on an inclined plane}
\author{Nivaldo A.  \surname{Lemos}} \email{nivaldolemos@id.uff.br}
\affiliation{Instituto de Física, Universidade Federal Fluminense,
Campus da Praia Vermelha, Niterói, 24210-340, RJ, Brazil.}

\date{\today}


\date{\today}

\begin{abstract}
\noindent Vakonomic mechanics has been proposed as a possible description of the dynamics of systems subject to nonholonomic  constraints. The aim of  the present work is to show that for an important physical system the motion brought about by vakonomic mechanics is completely inequivalent to the one derived from nonholonomic mechanics, which relies on  the standard method of Lagrange multipliers in the d'Alembert-Lagrange formulation of the classical equations of motion. For the rolling coin on an inclined plane, it is proved that no nontrivial solution to the equations of motion  of nonholonomic mechanics can be obtained in the framework of vakonomic mechanics. This completes previous investigations that managed to show only that, for certain mechanical systems, some but not necessarily all  nonholonomic motions are beyond the reach of vakonomic mechanics. Furthermore, it is argued that a simple qualitative experiment that anyone can perform at home supports the predictions of nonholonomic mechanics.  
\end{abstract}

\maketitle

\section{Introduction}

The study of nonholonomic mechanical systems is more than a hundred and fifty years old, with an eventful  and captivating history \cite{Leon,Borisov2}.  Their many applications and intriguing aspects  explain the ongoing  interest of physicists, mathematicians and engineers in a seemingly outmoded subject.
The important distinction between holonomic and nonholonomic constraints  appears to have first been made by Hertz \cite{Hertz}, who also  coined the 
terms ``holonomic'' and ``nonholonomic''.

In classical mechanics one often encounters problems involving   rolling without slipping, such as that of a rolling coin on an inclined plane  \cite{Lemos} or  two wheels mounted on a rigid axle \cite{Wang}, and also those entailing  the requirement  of not skidding sideways, as in the case of  a skate 
\cite{Lemos,Neimark}, a bicycle \cite{Lemos,Neimark} or  a snakeboard
 \cite{Janova}.  As to applications in engineering,   systems subject to similar  constraints play a major role in the thriving field of robotics \cite{Novel,Seda}. 

The  restrictions that characterize the allowed motions of  systems such as those mentioned above are   set forth as velocity-dependent constraints. The equations that mathematically express these constraints are first-order ordinary differential equations that depend linearly on the velocities associated with the coordinates chosen to describe the configurations of the system. In virtually all cases these constraints are nonholonomic, and they are usually tackled by the method of Lagrange multipliers together with an appeal to virtual displacements and Hamilton's principle. The resulting equations of motion define the d’Alembert-Lagrange dynamics of nonholonomic systems or simply  ``nonholonomic mechanics''. 

The d’Alembert-Lagrange  mathematical model remained  unchallenged until about 40 years ago, when an alternative dynamics for  mechanical systems with velocity-dependent constraints  was advocated   
 by Kozlov \cite{Kozlov1,Kozlov} and endorsed by other Russian mathematicians \cite{Arnold}. 
This mechanics arises from the  assumption that the problem
of finding the equations of motion for nonholonomic systems is the same as the problem of Lagrange in the calculus of variations \cite{Bliss,Bliss2}. As a shorthand for ``mechanics of variational axiomatic kind'', 
the name ``vakonomic mechanics''  was created by
Kozlov for his way of viewing the dynamics of nonholonomic systems. Another common name for vakonomic mechanics is ``variational nonholonomic mechanics''.

During the last few decades, both   nonholonomic and vakonomic mechanics have been vigorously investigated, mostly by mathematicians, and often with the use of sophisticated    
geometrical methods \cite{Leon,Cardin,Borisov,Cortes1}. An interesting feature of vakonomic mechanics is that it can be cast in Hamiltonian form \cite{Arnold} and treated by the Dirac theory of constrained Hamiltonian systems \cite{Cortes2}.  Nonstandard transpositional relations between variation of derivative and derivative  of variation lead to a modified vakonomic mechanics \cite{LLibre}. Another topic that has been pursued, chiefly by mathematicians,  is the reduction of nonholonomic systems with symmetry, which consists in
 obtaining equations of motion with fewer coordinates than those one started with \cite{Koiller}.

Of course, a pressing  issue is that of the equivalence of the nonholonomic and vakonomic approaches to the dynamics of nonholonomic systems.  Sufficient conditions have been found under which the vakonomic and nonholonomic equations of motion for nonholonomic systems can be regarded as equivalent \cite{Cortes1,Favretti,Fernandez}. In spite of these mathematical results, the proposal of vakonomic mechanics to describe nonholonomic systems has been criticized  on physical grounds, from  theoretical as well as from  empirical viewpoints \cite{Kharlamov,Li,Zampieri,Lewis,Kai}.  On the theoretical side, it has been argued that vakonomic mechanics contains ``extra parameters'' with no mechanical meaning  and that Kozlov's ideas on the realization of constraints are open to doubt \cite{Kharlamov}; that  in the vakonomic treatment of
general nonholonomic systems the nonvanishing virtual work of constraint 
forces is unjustifiably omitted \cite{Li}; and that the vakonomic skate on an inclined plane displays unphysical behavior \cite{Zampieri}.  For the nonholonomic system consisting of a ball rolling on a rotating table, theoretical and experimental studies indicate that at least some actual motions performed by the ball cannot be brought about by  vakonomic mechanics \cite{Lewis,Kai}. 
 
To our knowledge, all investigations so far have only shown that for some nonholonomic systems vakonomic mechanics is partially inequivalent to  nonholonomic 
mechanics, meaning that some but not necessarily all nonholonomic motions are not vakonomic motions. Here we prove that, for  a rolling coin on an inclined plane, every nontrivial nonholonomic motion is beyond the reach of vakonomic  mechanics. Besides, we argue that a simple qualitative experiment that anyone can perform at home supports the predictions of nonholonomic mechanics. 

This paper is organized as follows. In Section \ref{Nonholonomic} we review the derivation of the d’Alembert-Lagrange dynamics of nonholonomic systems by the method of Lagrange multipliers. In Section \ref{Vakonomic} the vakonomic equations of motion are derived by the same variational technique as pertains to the problem of Lagrange in the calculus of variations, and some general physical features of vakonomic mechanics are discussed. In Section \ref{rollcoin} we set up the nonholonomic and vakonomic  equations 
of motion for a rolling coin on an inclined plane, and prove that vakonomic mechanics is unable to give rise to any nontrivial nonholonomic motion. Section \ref{Conclusions} is dedicated to  conclusions and a few broad remarks.


\section{Nonholonomic mechanics}\label{Nonholonomic}

In order to make our discussion as self contained as possible, let us briefly review the traditional derivation of the equations of motion of nonholonomic mechanics by the method of Lagrange mutipliers when the constraints are expressed as first-order linear differential equations for the coordinates. This approach relies in a crucial way on the notion of virtual displacement, which is an essential ingredient in d'Alembert's principle \cite{Lemos,Pars,Goldstein}.

Consider an $n$-degree-of-freedom mechanical system described by coordinates $q_1, \ldots , q_n$ subject to a set of the $p<n$ mutually independent differential constraints 
\begin{equation}
\label{vinculos-lineares-velocidades-diferenciais}
\sum_{k=1}^n\, a_{lk} dq_k  +  a_{lt}dt  = 0\, , \,\,\,\,\,\,\,\,\,\,
l=1,\ldots ,p\, .
\end{equation}
The coefficients $a_{lk}$ and $a_{lt}$ do not depend on the velocities, that is, they are functions of $q_1, \ldots , q_n, t$ alone.
In order to ensure the mutual independence of the above constraint equations  we assume that the $p \times n$  matrix $\boldsymbol{\mathsf{A}} = (a_{lk})$
has rank $p$ for all $q,t$. 

Let $L(q, {\dot q},t)$ be the Lagrangian for the system, {\it written as if there were no constraints}. The standard Lagrange multiplier method to take into account the constraints (\ref{vinculos-lineares-velocidades-diferenciais}) in Hamilton's principle \cite{Lemos} starts by noting that the variations $\delta q_k$ that enter Hamilton's variational principle   are virtual displacements that, because of
(\ref{vinculos-lineares-velocidades-diferenciais}), obey
\begin{equation}
\label{vinculos-lineares-velocidades-virtual}
\sum_{k=1}^n\, a_{lk} \delta q_k   = 0\, , \,\,\,\,\,\,\,\,\,\,
l=1,\ldots ,p\, 
\end{equation}
inasmuch as $dt =0$ for virtual displacements. 

On the assumption that Hamilton's principle $\delta S = 0$ remains  valid for nonholonomic systems, one writes
\begin{equation}
\label{deltaS=zero-nao-holonomo}
\delta S = \delta \int_{t_1}^{t_2}L dt = \int_{t_1}^{t_2}\, dt\, \sum_{k=1}^n \Biggl [ \frac{\partial L}{\partial q_k}
 -\frac{d}{dt}
\Bigl(\frac{\partial L}{\partial  {\dot q_k}}\Bigr) \Biggr]
\delta  q_k =0
\, ,
\end{equation}
where $\delta  q_k(t_1)= \delta  q_k(t_2)=0$. The statement that $\delta S = 0$ for all variations $\delta  q_k$ that satisfy (\ref{vinculos-lineares-velocidades-virtual}) is known as the d'Alembert-Lagrange principle.     

 Since the $\, \delta q$s are not mutually independent because of (\ref{vinculos-lineares-velocidades-virtual}),  one  cannot infer that the coefficient of each
$\delta q_k$ is zero.
In order to circumvent this difficulty one adds to   $\delta S = 0$ a special zero made up of  the left-hand-side of (\ref{vinculos-lineares-velocidades-virtual}) multiplied by an arbitrary function of time, a Lagrange multiplier ${\lambda}_l$, and summed over $l$. This procedure yields
\begin{equation}
\label{deltaS=zero-com-multiplicadores-de-Lagrange}
\int_{t_1}^{t_2}\, dt\, \sum_{k=1}^n \Biggl\{
 \frac{\partial L}{\partial q_k}
 -\frac{d}{dt}
\Bigl(\frac{\partial L}{\partial  {\dot q_k}}\Bigr)  + 
\sum_{l=1}^p\,{\lambda}_l a_{lk}  \Biggr\}
\delta  q_k =0
\, .
\end{equation}

Since  $\boldsymbol{\mathsf{A}}$ has rank $p$, with a suitable numbering of the coordinates  the first $p$ variations  $\delta q_1, \dots , \delta q_p$ can be expressed in terms of the remaining ones $\delta q_{p+1}, \dots , \delta q_n$, which are independent and arbitrary. Also because the rank of  $\boldsymbol{\mathsf{A}}$ is $p$, it is possible  to pick the $p$ Lagrange multipliers $\lambda_1(t), \ldots , \lambda_p(t)$ so   that the coefficient of each of the first $p$ variations $\delta q_1, \dots , \delta q_p$ in equation (\ref{deltaS=zero-com-multiplicadores-de-Lagrange})  is  zero. Having made this choice, and taking into account that the remaining $\delta q$s are independent and arbitrary, Hamilton's principle 
leads to the equations of motion \cite{Lemos,Whittaker}   
\begin{equation}
\label{equacoes-Lagrange-com-multiplicadores-de-Lagrange}
\frac{d}{dt}
\Bigl(\frac{\partial L}{\partial  {\dot q_k}} \Bigr) -
\frac{\partial L}{\partial q_k} =
\sum_{l=1}^p\,{\lambda}_l a_{lk}
\, , \,\,\,\,\,\,\,\,\,\,
k=1,\ldots ,n .
\end{equation}

Of course, the differential constraints (\ref{vinculos-lineares-velocidades-diferenciais}) can be written in the equivalent form of a set of $p$ first-order ordinary differential equations:  
\begin{equation}
\label{vinculos-lineares-velocidades-equacoes-diferenciais}
\sum_{k=1}^n\, a_{lk} {\dot q}_k  +  a_{lt}  = 0\, , \,\,\,\,\,\,\,\,\,\,
l=1,\ldots ,p\, .
\end{equation}

Equations  (\ref{equacoes-Lagrange-com-multiplicadores-de-Lagrange}) and (\ref{vinculos-lineares-velocidades-equacoes-diferenciais}) comprise a set  of $n+p$
equations for $n+p$ unknowns --- the $n$ coordinates $q_1, \ldots , q_n$ and the $p$ Lagrange multipliers $\lambda_1, \ldots , \lambda_p$ --- which under suitable conditions  uniquely determine the motion of the system. For a natural Lagrangian $L = T - V$, the potential energy $V$ depends only on the coordinates and the kinetic energy $T$ is a positive quadratic form in the velocities, hence the Hessian matrix $\boldsymbol{\mathsf{W}}= (\partial^2 L/\partial {\dot q}_i\partial {\dot q}_j)$ is positive. Then it follows \cite{Moriconi} that  equations (\ref{equacoes-Lagrange-com-multiplicadores-de-Lagrange}) and (\ref{vinculos-lineares-velocidades-equacoes-diferenciais}) can be uniquely solved for the Lagrange multipliers and accelerations in the form $\lambda_l = \mu_l(q, {\dot q},t)$,
${\ddot q}_k = u_k(q, {\dot q},t)$. As a consequence, the motion $q(t)$ is uniquely determined once initial conditions satisfying the constraints are stipulated. 
 
According to Whittaker \cite{Whittaker}, the above extension  of Hamilton's  principle to nonholonomic systems was first made by H\"older in 1896. One significant aspect of this derivation of the equations of motion based on an adaptation of Hamilton's principle  is that only the physical path obeys the constraints: the varied paths do not satisfy the constraints unless the constraints are holonomic \cite{Pars,Whittaker}. Besides,  the d'Alembert-Lagrange principle is not an actual variational principle because, except for  the case in which the  constraints are holonomic, there is no functional  ${\cal S} =\int_{t_1}^{t_2} {\cal L}(q, {\dot q},t) dt$ such that equations (\ref{equacoes-Lagrange-com-multiplicadores-de-Lagrange}) arise from the stationarity condition $\delta {\cal S} =0$.

The technique of Lagrange multipliers to derive equations of motion can be extended to nonholonomic  constraints of the more general  form
\begin{equation}
\label{non-holon-constraints-more-general}
g_l(q, {\dot q}, t) =0 \, , \,\,\,\,\,\,\,\,\,\,
l=1,\ldots ,p
\end{equation}
where $g_l: \mathbb{R}^{2n+1} \to \mathbb{R}$ are infinitely differentiable functions.
In this case, the previous arguments based on virtual displacements no longer work.
Obtained either from Hertz's principle of least curvature \cite{Rund}, from  a modified variational principle \cite{Saletan}, or from new transpositional
relations \cite{Flannery2}, to name a few approaches,
the accepted equations of motion are 
\begin{equation}
\label{nonholonomic-equations-of-motion}
\frac{d}{dt}
\Bigl(\frac{\partial L}{\partial  {\dot q_k}} \Bigr) -
\frac{\partial L}{\partial q_k} =
\sum_{l=1}^p\,{\lambda}_l \frac{\partial g_l}{\partial {\dot q}_k}
\, , \,\,\,\,\,\,\,\,\,\,
k=1,\ldots ,n\, .
\end{equation}
It is clear that these  reduce to the previous equations of motion (\ref{equacoes-Lagrange-com-multiplicadores-de-Lagrange}) when the functions $g_l$ depend linearly on the velocities.  Now we assume that the rectangular  matrix $\boldsymbol{\mathsf{A}}$ with elements  $a_{lk} = \partial g_l/ \partial {\dot q}_k$  has rank $p$ for all $q, {\dot q}, t$.  

\begin{Def}
A {\bf nonholonomic motion} is any solution $q(t)$ to equations  (\ref{non-holon-constraints-more-general}) and (\ref{nonholonomic-equations-of-motion}).
\end{Def}

\subsection{Constraint forces}

Given a system composed of $N$ particles, Newton's second law expressed in terms of arbitrary and independent coordinates $q_1, \ldots , q_n$ reads \cite{Lemos}
\begin{equation}
\label{eqs-Lagrange-emtermos-T-Qk}
\frac{d}{dt} \Bigl (
\frac{\partial T}{\partial {\dot q}_k} \Bigr )
- \frac{\partial T}{\partial  q_k}    =  Q_k \, , \,\,\,\,\,\,\,\,\,
k=1, \ldots ,n \, , 
\end{equation}
where  $n=3N$ and
\begin{equation}
\label{forcas-generalizadas}
Q_k =  \sum_{i=1}^N {\bf F}_i \, \cdot \frac{\partial {\bf r}_i}{\partial q_k}.
\end{equation}
$Q_k$ is, by definition, the $k$-th component of the   generalized force. In the last equation, ${\bf r}_i = {\bf r}_i(q,t)$ is the position vector of the $i$-th particle and  ${\bf F}_i$ is the total force acting on it.

The generalized force can be written as sum of a part that arises from a generalized potential $U(q,{\dot q},t)$ with  another part that cannot be derived from any generalized potential:
\begin{equation}
\label{forcas-potencial-generalizado-mais-nao-conservativas}
Q_k = - \frac{\partial U}{\partial q_k} + \frac{d}{dt} \Bigl (
\frac{\partial U}{\partial {\dot q}_k} \Bigr ) + Q_k^{\prime}. 
\end{equation}
With this decomposition, equation (\ref{eqs-Lagrange-emtermos-T-Qk}) becomes
\begin{equation}
\label{equacoes-de-Lagrange-mais-forcas-nao-conservativas}
\frac{d}{dt} \Bigl (
\frac{\partial L}{\partial {\dot q}_k} \Bigr )
- \frac{\partial L }{\partial  q_k}  =  Q_k^{\prime} \, , 
\end{equation}
with $L = T - U$. 
 
Since the generalized potential $U$ in the Lagrangian $L = T - U$ has nothing to do with the constraint forces, but refers to the applied forces alone, a comparison
of equations (\ref{nonholonomic-equations-of-motion}) and (\ref{equacoes-de-Lagrange-mais-forcas-nao-conservativas}) shows that the right-hand side of equation
(\ref{nonholonomic-equations-of-motion}) represents the constraint forces. Thus, 
in terms of the Lagrange multipliers, which are determined in the process of solving equations  (\ref{non-holon-constraints-more-general}) and (\ref{nonholonomic-equations-of-motion}),  the generalized constraint forces $Q^{\prime}_k$  responsible for the enforcement of (\ref{non-holon-constraints-more-general}) are given by 
\begin{equation}
\label{forcas-de-vinculo-em-termos-multiplicadores-Lagrange}
Q^{\prime}_k= \sum_{l=1}^p\,{\lambda}_l \frac{\partial g_l}{\partial {\dot q}_k} \, , \,\,\,\,\,\,\,\,\,\,
k=1,\ldots ,n\,  .
\end{equation}

\subsection{Virtual displacements}

In order for the d'Alembert-Lagrange formulation of the classical equations of motion for nonholonomic systems to be made compatible with the principle of virtual work, the total virtual work of the constraint forces must vanish, that is,
one must have
\begin{equation}
\label{virtual-work-zero}
 \sum_{k=1}^n\,Q^{\prime}_k \delta q_k =0
\end{equation}
for all  virtual displacements $\delta q_1, \ldots ,\delta q_n$. 

This means that for general constraints of the form (\ref{non-holon-constraints-more-general}) the  virtual displacements are to be restricted by \cite{Arnold,Saletan}
\begin{equation}
\label{virtual-displacements-general}
 \sum_{k=1}^n\,\frac{\partial g_l}{\partial {\dot q}_k} \delta q_k =0,  \,\,\,\,\,\,\,\,\,\, 
l=1,\ldots ,p .
\end{equation}
This condition ensures that the constraint forces (\ref{forcas-de-vinculo-em-termos-multiplicadores-Lagrange}) obey  equation (\ref{virtual-work-zero}) no matter what the values of the Lagrange multipliers may be. Furthermore, equation (\ref{virtual-displacements-general}) reduces to  (\ref{vinculos-lineares-velocidades-virtual}) when the constraints depend linearly on the velocities.

\section{Vakonomic mechanics}\label{Vakonomic}

It has long been known that the {\it physical} Lagrange multiplier method to deal with nonholonomic constraints  in classical mechanics 
that was succintly described above is not equivalent to the treatment of the corresponding {\it mathematical} problem of Lagrange  in the calculus of variations \cite{Bliss,Bliss2}. Applied to nonholonomic systems,  the latter gives rise to  vakonomic mechanics \cite{Ray,Goedecke,Kozlov1,Kozlov,Arnold},  which prominently features ${\dot \lambda_l}$   in the equations of motion in addition to the Lagrange multipliers $\lambda_l$ themselves. The equations of motion of vakonomic mechanics arise from the variational principle
 \begin{equation}
\label{vakonomic-variational-principle}
\delta \int_{t_1}^{t_2}\Bigl[ L(q, {\dot q},t) + \sum_{l=1}^p \lambda_l(t) g_l(q, {\dot q},t) \Bigr] dt = 0  ,
\end{equation} 
in which the $q$s and $\lambda$s are varied independently, with the standard boundary conditions $\delta q_k(t_1) =  \delta q_k(t_2) = 0$. Variation of the Lagrange multipliers  leads to the constraint equations (\ref{non-holon-constraints-more-general}), whereas variation of the coordinates yields the equations of motion 
\begin{equation}
\label{vakonomic-equations-of-motion}
\frac{d}{dt}
\Bigl(\frac{\partial L}{\partial  {\dot q_k}} \Bigr) -
\frac{\partial L}{\partial q_k} =
\sum_{l=1}^p\,{\lambda}_l \bigg[ \frac{\partial g_l}{\partial q_k} - \frac{d}{dt}
\Bigl(\frac{\partial g_l}{\partial  {\dot q_k}} \Bigr) \bigg] - \sum_{l=1}^p {\dot \lambda_l}\frac{\partial g_l}{\partial  {\dot q_k}}
\, , \,\,\,\,\,\,\,\,\,\,
k=1,\ldots ,n\, .
\end{equation}
These vakonomic equations of motion differ significantly from the nonholonomic equations of motion (\ref{nonholonomic-equations-of-motion}), most notably because of the  conspicuous appearence of the time derivative of the Lagrange multipliers in   addition to the Lagrange multipliers themselves. This means that in order to get a unique solution for the path $q(t)$ one must prescribe not only the initial values
of the coordinates and velocities but also  initial values ${\lambda}_1(0), \ldots , {\lambda}_p(0)$ for the  Lagrange multipliers. 

\begin{Def}
A {\bf vakonomic motion} is any solution $q(t)$ to equations  (\ref{non-holon-constraints-more-general}) and (\ref{vakonomic-equations-of-motion}).
\end{Def}

\begin{Def}
Vakonomic mechanics is said to be {\bf completely inequivalent} to nonholonomic mechanics for some mechanical system if the sets of vakonomic motions and nonholonomic motions are  disjoint.
\end{Def}

The first important distinction between the vakonomic variational principle and the d'Alembert-Lagrange extension of Hamilton's principle discussed in Section \ref{Nonholonomic} is that in the  vakonomic variational principle not only the actual path but also the comparison paths obey the constraints. The second important difference is that  equations (\ref{vakonomic-equations-of-motion}) arise from the genuine  variational principle $\,\delta\int_{t_1}^{t_2} {\cal L}(q, {\dot q},t) dt =0\,$ where
$\,{\cal L} = L + \sum_{l} \lambda_l g_l\,$ with $\lambda_1, \ldots , \lambda_p$ regarded as additional coordinates, and  in which the $q$s and $\lambda$s
 are varied independently.

For holonomic constraints
\begin{equation}
\label{holonomic-constraints}
f_l(q,t) =0, \,\,\,\,\,\,\,\,\,\, l=1,\ldots ,p,
\end{equation}
we have 
\begin{equation}
\label{g-holonomic-constraints}
g_l(q, {\dot q},t) = \frac{d}{dt} f_l(q,t) = \sum_{j=1}^n \frac{\partial f_l}{\partial q_j}{\dot q}_j + \frac{\partial f_l}{\partial t},
\end{equation}
and it is easily  checked  that
\begin{equation}
\label{g-variational-derivative-zero}
 \frac{d}{dt}
\Bigl(\frac{\partial g_l}{\partial  {\dot q_k}}\Bigr) -\frac{\partial g_l}{\partial q_k}  =0.
\end{equation}
This reduces the vakonomic equations of motion  (\ref{vakonomic-equations-of-motion}) to
\begin{equation}
\label{vakonomic-equations-of-motion-holonomic}
\frac{d}{dt}
\Bigl(\frac{\partial L}{\partial  {\dot q_k}} \Bigr) -
\frac{\partial L}{\partial q_k} =
 - \sum_{l=1}^p {\dot \lambda_l}\frac{\partial g_l}{\partial  {\dot q_k}}
\, , \,\,\,\,\,\,\,\,\,\,
k=1,\ldots ,n\, .
\end{equation}
Since only ${\dot \lambda_l}$ appear in these equations but not the Lagrange multipliers $ \lambda_l$ themselves, the mere renaming  ${\dot \lambda_l} \to -\lambda_l$ shows that the equations of motion of  vakonomic mechanics coincide with (\ref{nonholonomic-equations-of-motion}) if the constraints are holonomic. The converse of this statement  also holds \cite{Rund,Lewis}. Therefore,
the set of vakonomic motions is equal  to the  set of nonholonomic motions if and only if the constraints are holonomic.

Just like Sgt. Pepper's band, equations (\ref{vakonomic-equations-of-motion}) have been going in and out of style. Although they had appeared in the last decade of the nineteenth century in the works of Hertz, H\"older and others \cite{Arnold}  in connection with the issue  of  applicability of Hamilton's principle to nonholonomic systems,    throughout the twentieth century the vakonomic equations of motion  never made into mainstream textbooks. They do not appear even in advanced texts on analytical mechanics such as Whittaker's authoritative treatise \cite{Whittaker}, in which only constraints that depend linearly on the velocities are considered, with the dynamics  governed by  (\ref{equacoes-Lagrange-com-multiplicadores-de-Lagrange}). Equations (\ref{vakonomic-equations-of-motion}) were rediscovered in the nineteen sixties \cite{Ray,Goedecke}, soon to be  disowned by one of its proponents \cite{Ray2}. Having gone down into obscurity, in the nineteen eighties they came to the fore again in the work of Kozlov \cite{Kozlov1,Kozlov,Arnold}, who came up with the name ``vakonomic mechanics'' as a shorthand for ``mechanics of variational axiomatic kind''. At the dawn of the twenty-first century, equations (\ref{vakonomic-equations-of-motion}) were sanctioned in the third edition of Goldstein's acclaimed textbook \cite{Goldstein}, but retracted without explanation by the new co-authors almost as soon as the book came to light \cite{Goldstein2}.

As discussed in the previous section, according to vakonomic mechanics the constraint forces should be given by the right-hand side of Eq. (\ref{vakonomic-equations-of-motion}), namely
\begin{equation}
\label{vakonomic-constraint-forces}
{Q^{\prime}_k}^{vak} =
\sum_{l=1}^p\,{\lambda}_l \bigg[ \frac{\partial g_l}{\partial q_k} - \frac{d}{dt}
\Bigl(\frac{\partial g_l}{\partial  {\dot q_k}} \Bigr) \bigg] - \sum_{l=1}^p {\dot \lambda_l}\frac{\partial g_l}{\partial  {\dot q_k}}.
\end{equation}
It is clear that, with the virtual displacements defined by (\ref{virtual-displacements-general}), the total virtual work of these vakonomic constraint forces does not vanish. 
Sometimes the second term alone on the right-hand side of the above  equation, the one containing the ${\dot \lambda}$s, is interpreted as the constraint force. Then the $\lambda$s themselves are regarded as free parameters that can be
conveniently chosen in order to force the system to follow a
prescribed path in configuration space \cite{Antunes}. This interpretation would make vakonomic mechanics compatible with the principle of virtual work, but it does
not seem to be tenable for at least two reasons. First, one would have to unmotivatedly abandon Newtonian mechanics and replace it by a new mechanics containing hidden parameters, namely the initial values of the Lagrange multipliers. Second, this interpretation suggests that by appropriately choosing these  hidden parameters one can make the system follow any prescribed path, which is not true, as  shown in this paper.

If  the constraints are nonholonomic, equations   (\ref{vakonomic-equations-of-motion}) and (\ref{nonholonomic-equations-of-motion}) 
are not equivalent and their respective sets of solutions are  different. This does not necessarily mean, however, that each and every nonholonomic motion  is beyond the reach of vakonomic mechanics. For some choice of the initial values of the Lagrange multipliers, the corresponding  vakonomic motion  might coincide with  some  nonholonomic motion. Mathematical conditions are known 
\cite{Cortes1,Favretti,Fernandez} under which equations (\ref{vakonomic-equations-of-motion}) and (\ref{nonholonomic-equations-of-motion}) can be regarded as equivalent for certain initial conditions and initial values of the Lagrange multipliers. In the following we show that these equivalence conditions do not cover at least one important physical system, namely the rolling coin on an inclined plane.

\section{Rolling coin on an inclined plane}\label{rollcoin}

\begin{figure}[t!]
\epsfxsize=10cm
\begin{center}
\leavevmode
\epsffile{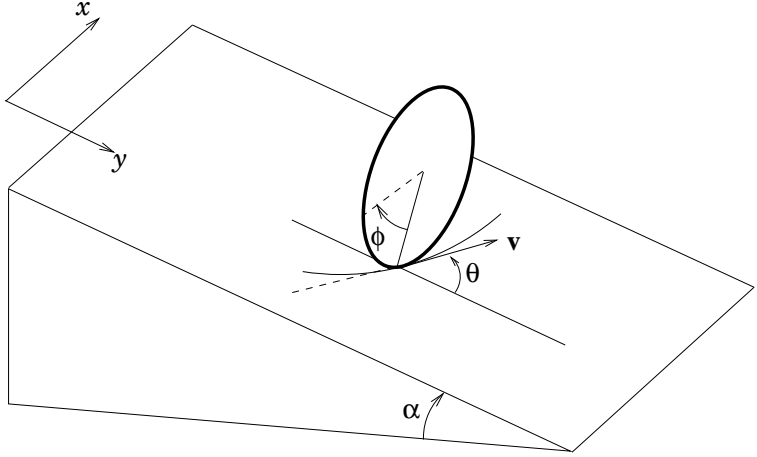}
\caption{Upright coin rolling without slipping on an inclined plane. The $x$-axis is horizontal and $\bf v$ is the center-of-mass velocity. Figure reproduced from the author's book, Ref. \cite{Lemos}.}
\label{moedaplanoinclinado}
\end{center}
\end{figure}

Let us consider \cite{Lemos} a coin of mass $\, m\,$ and radius $\, R\,$ that rolls on an  inclined plane always remaining upright --- its plane is always perpendicular to the inclined plane (see  Fig. \ref{moedaplanoinclinado}). Principal axes through the center of mass are any two orthogonal axes lying in the plane of the coin, with the third axis perpendicular to the plane of the coin. An elementary calculation gives  $\, I_3 =mR^2/2\,$ and, by symmetry and the perpendicular axis theorem,  $\, I_1=I_2=I_3/2=mR^2/4\,$. The kinetic energy  is  the center-of-mass  translational energy plus the kinetic energy of rotation about the center of mass. Therefore, the Lagrangian is  
\begin{equation}
\label{lagrangiana-moeda-rolante-em-termos-eixos-principais-inercia}
L= T-V= \frac{m}{2}({\dot x}^2 +{\dot y}^2) + \frac{1}{2}(I_1\omega_1^2 + I_2\omega_2^2 + I_3\omega_3^2) 
+ mgy\sin\alpha\, , 
\end{equation}
where $\, \alpha\,$ is the  angle of the slope of the inclined plane  and $\, x,y\,$ are Cartesian coordinates of the center of mass of the coin (note that the $x$-axis is horizontal). Let $\, \phi\,$ and
$\, \theta \,$ be angles describing the  rotations of the coin about its symmetry axis parallel to the inclined plane (the third principal axis) and an axis perpendicular to the inclined plane,  respectively, as shown in  Fig. \ref{moedaplanoinclinado}. Consequently, $\,\omega_3={\dot\phi}\,$
and, taking account that   $\,{\dot\theta}\,$ is the projection of the angular velocity  vector on the  plane of the coin,  $\,\omega_1^2+\omega_2^2={\dot\theta}^2$. Thus, the Lagrangian for the system takes the form

\begin{equation}
\label{lagrangiana-moeda-rolante}
L= \frac{m}{2}({\dot x}^2 +{\dot y}^2) + \frac{mR^2}{4} {\dot\phi}^2 + 
\frac{mR^2}{8} {\dot\theta}^2 + mgy\sin\alpha\, . 
\end{equation}

In Fig. \ref{moedaplanoinclinado} the center-of-mass velocity $\bf v$ is depicted at the point of contact between the coin and the table to emphasize that it is tangent to the path traced out by the contact point on the table. The instantaneous velocity  of the contact point is zero, which implies $v = R {\dot \phi}$. Since ${\dot x}=v\sin \theta$ and ${\dot y}=v\cos \theta$, the  rolling constraints are 
\begin{subequations}
\label{constraints}
\begin{align}
\label{constraints-x}
{\dot x} - R {\dot \phi}\sin \theta = 0  , \\
\label{constraints-y}
{\dot y} - R {\dot \phi}\cos \, \theta = 0 . 
\end{align}
\end{subequations}
By means of the Frobenius theorem, it can be readily shown that these constraints are nonholonomic \cite{Lemos}.

\subsection{Rolling coin dynamics according to nonholonomic mechanics}

Setting $q=(x,y, \phi, \theta )$, from  Eqs. (\ref{constraints}) one immediately identifies the two constraint functions  as
\begin{subequations}
\label{constraint-functions}
\begin{align}
\label{constraint-functions1}
g_1(q,{\dot q},t) = {\dot x} - R {\dot \phi}\sin \theta , \\
\label{constraint-functions-2}
g_2(q,{\dot q},t) = {\dot y} - R {\dot \phi}\cos \, \theta  . 
\end{align}
\end{subequations}
Therefore, the nonholonomic equations of motion (\ref{nonholonomic-equations-of-motion}) take the form 
\begin{eqnarray}
 m{\ddot x} & = & \lambda_1\, ,\label{moeda-rolante-equacao-movimento-a}\\
 m{\ddot y} & = & mg\sin\alpha + \lambda_2\, ,\label{moeda-rolante-equacao-movimento-b}\\
 \frac{mR^2}{2}\, {\ddot \phi} & = & -\lambda_1 R\sin\theta - \lambda_2 R \cos\theta .
\label{moeda-rolante-equacao-movimento-d} \\
 \frac{mR^2}{4}\, {\ddot \theta}  & = & 0, \label{moeda-rolante-equacao-movimento-c}
\end{eqnarray}

From (\ref{moeda-rolante-equacao-movimento-c}) it  immediately follows that  
\begin{equation}
\label{moeda-rolante-solucao-theta}
\theta =  \Omega \, t + \theta_0
\end{equation}
where  $\, \Omega \,$ and $\, \theta_0\,$ are arbitrary  constants.
Combining equations (\ref{constraints-x}) and (\ref{moeda-rolante-equacao-movimento-a}) we get 
\begin{equation}
\label{moeda-rolante-lambda1}
\lambda_1 = mR{\ddot \phi}\sin\theta + mR \Omega {\dot \phi} \cos\theta \, ,
\end{equation}
and, proceeding similarly,
\begin{equation}
\label{moeda-rolante-lambda2}
\lambda_2 = mR{\ddot \phi} \cos\theta - mR \Omega {\dot \phi}\sin\theta  - mg 
\sin\alpha
\, .
\end{equation}
Substitution of these expressions for  $\, \lambda_1\,$ and 
$\, \lambda_2\,$ into (\ref{moeda-rolante-equacao-movimento-d}) leads to the following differential equation for  $\, \phi$:
\begin{equation}
\label{moeda-rolante-equacao-para-phi}
{\ddot \phi}= \frac{2g \sin\alpha}{3R} 
\cos ( \Omega \, t + \theta_0 )
\, .
\end{equation}
The general solution to this equation is 
\begin{equation}
\label{moeda-rolante-solucao-para-phi}
 \phi= \phi_0 + \omega t - \frac{2g \sin\alpha}{3\Omega^2 R} 
\cos ( \Omega \, t + \theta_0 ),
\end{equation}
where, of course, we assume that $\Omega \neq 0$.
Finally, substituting (\ref{moeda-rolante-solucao-theta}) and  (\ref{moeda-rolante-solucao-para-phi}) into the constraint  equations  (\ref{constraints-x}) and (\ref{constraints-y}) and integrating, we find
\begin{eqnarray}
& & x= x_0 + \frac{g\sin\alpha}{3\Omega}\, t - \Biggl[ \frac{\omega 
R}{\Omega} + \frac{g\sin\alpha}{3\Omega^2} \sin(\Omega \, t + \theta_0)\Biggr] \, \cos ( \Omega \, t + \theta_0 )
\, , \label{moeda-rolante-solucao-para-x} \\
& &  y= y_0 +  \Biggl[ \frac{\omega 
R}{\Omega} + \frac{g\sin\alpha}{3\Omega^2} \sin ( \Omega \, t + \theta_0 )\Biggr] \sin ( \Omega \, t + \theta_0 )
\, . \label{moeda-rolante-solucao-para-y}
\end{eqnarray}
This completes the integration of the equations of motion for the upright rolling coin in terms of the six arbitrary constants $\, x_0,y_0,\theta_0,\phi_0,\omega,\Omega\,$. Note that the general solution to the equations of motion contains only six arbitrary constants because each nonholonomic constraint reduces
the number of degrees of freedom by one-half unit.

Surprisingly, equations (\ref{moeda-rolante-solucao-para-x}) and (\ref{moeda-rolante-solucao-para-y}) assert that  the coupling between translation and rotation causes the horizontal motion  of the center of mass of the coin to be  unbounded, but the motion along the slanting $\, y$-axis to be bounded and oscillating: if the inclined plane is sufficiently long the coin will never reach its lowest edge. 

Before moving on to the vakonomic dynamics, it is worthwhile to consider the nonholonomic motion associated with the following  initial conditions:
\begin{subequations}
\label{initial-conditions}
\begin{align}
\label{initial-conditions1}
x(0)= 0, \,\,\,\,\,\,\,\,\,\,  {\dot x}(0)= 0, \,\,\,\,\,\,\,\,\,\, y(0)= 0, \,\,\,\,\,\,\,\,\,\,  {\dot y}(0)= 0, \\
\label{initial-conditions2}
\phi (0)= 0, \,\,\,\,\,\,\,\,\,\,  {\dot \phi}(0)= 0, \,\,\,\,\,\,\,\,\,\, \theta (0)= \frac{\pi}{2}, \,\,\,\,\,\,\,\,\,\,  {\dot \theta}(0)= \Omega  . 
\end{align}
\end{subequations}
Clearly these initial conditions satisfy the constraints   (\ref{constraints}), 
and equations  (\ref{moeda-rolante-solucao-theta}), (\ref{moeda-rolante-solucao-para-phi}), (\ref{moeda-rolante-solucao-para-x}) and (\ref{moeda-rolante-solucao-para-y}) give rise to the following motion for the coin's center of mass:
\begin{eqnarray}
& & x(t) =  \frac{g\sin\alpha}{3\Omega^2}\, \bigl( \Omega t - 2 \sin \Omega  t + \sin \Omega  t\, \cos  \Omega  t \bigr)
\, , \label{moeda-rolante-solucao-para-x-initial-conditions} \\
& &  y(t) =\frac{g\sin\alpha}{3\Omega^2}\, \bigl( 1 - \cos  \Omega  t \bigr)^2 .
 \label{moeda-rolante-solucao-para-y-initial-conditions}
\end{eqnarray}
Fig. (\ref{coinpath}) shows two typical trajectories for the coin's center of mass predicted by equations (\ref{moeda-rolante-solucao-para-x-initial-conditions}) and (\ref{moeda-rolante-solucao-para-y-initial-conditions}).
\begin{figure}[b!]
    \centering
    \begin{minipage}{.5\textwidth}
        \centering
        \includegraphics[width=0.7\linewidth, height=0.3\textheight]{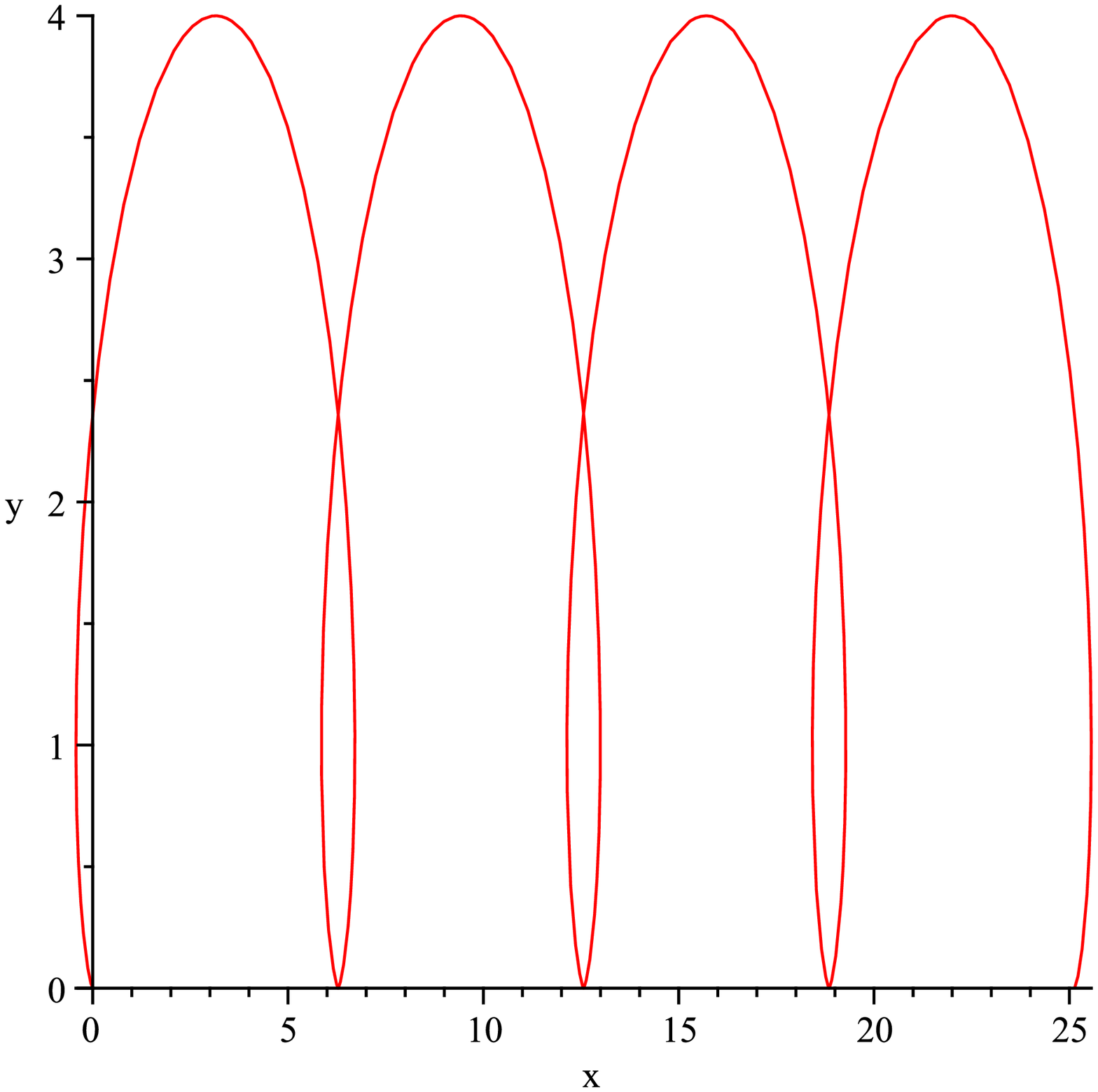}
    \end{minipage}%
    \begin{minipage}{0.5\textwidth}
        \centering
        \includegraphics[width=0.7\linewidth, height=0.3\textheight]{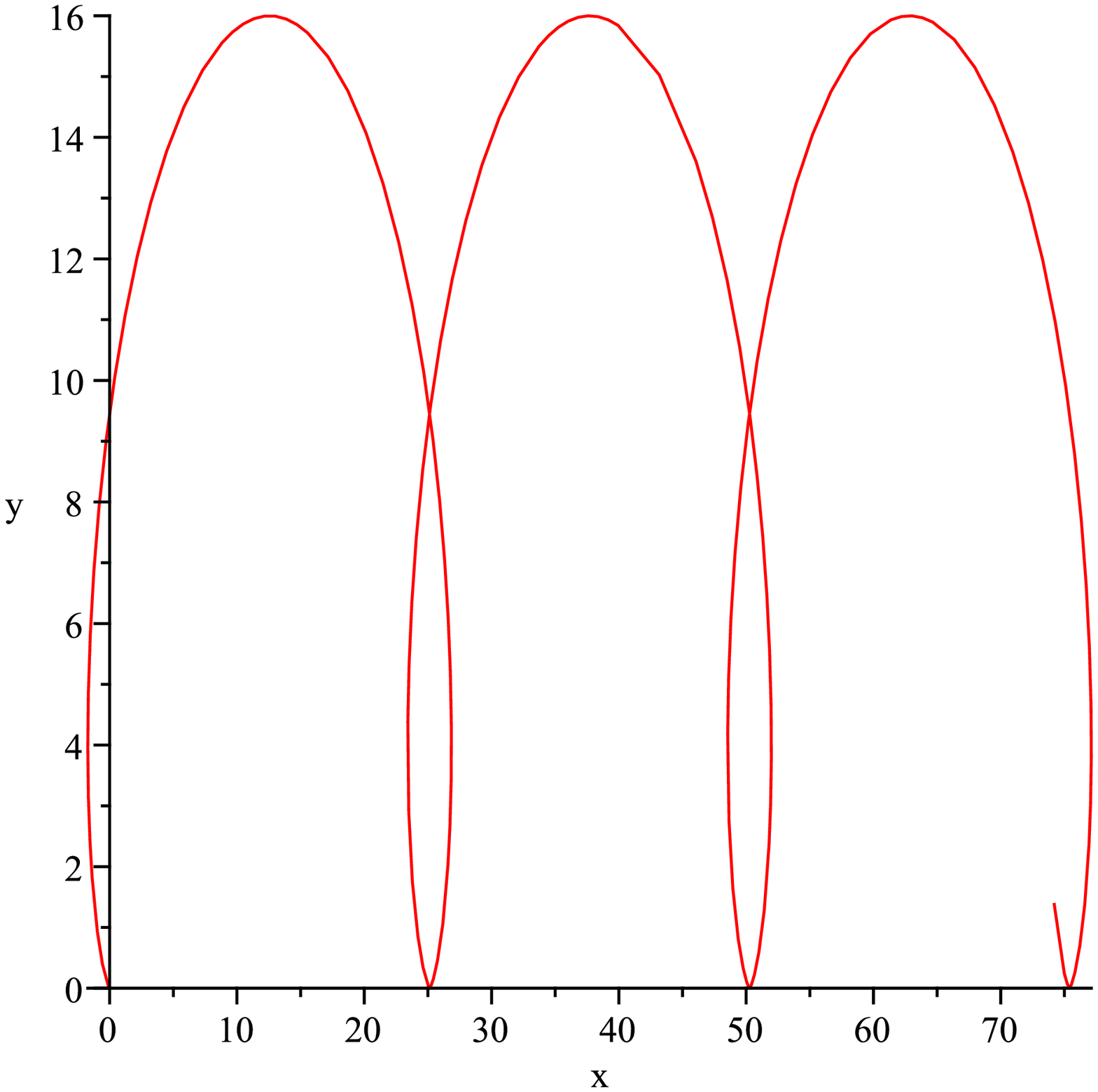}
    \end{minipage}
		\caption{Typical trajectories of the  rolling coin's center of mass according to nonholonomic mechanics for the initial conditions (\ref{initial-conditions}).   The angular velocity  for the left  graph is $\Omega$, whereas it is  $\Omega /2$ for the one on the right. The coordinates $x,y$ are depicted in units of $g \sin \alpha /3\Omega^2$. }
\label{coinpath}		
\end{figure}

The particular initial conditions (\ref{initial-conditions}) have been chosen because they are easy to realize in a qualitative home experiment. The first step consists in creating a gently inclined plane with the help of two books. 
Use  a finger of one hand to keep the coin upright, with its plane perpendicular to the slightly slanting edge of the inclined plane. Then release the coin at the same time as you flick its rim with a finger of the other hand, as in Fig. \ref{InitialCoin}. Qualitatively, the motion of the coin's center of mass resembles the ones shown in Fig. \ref{coinpath}, although a few trials may be needed to produce such a motion because the coin should be rolling without slipping from the start. One feature that is unequivocally borne out by this domestic experiment is that the coin does not fall all the way down, but reverses its downward motion and climbs back up the inclined plane, in violation of what intuition suggests. This behavior, which is also displayed by a skate \cite{Flannery3}, is sometimes regarded as paradoxical \cite{Arnold}, and invoked to   suggest  that nonholonomic mechanics leads to unphysical behavior. The paradox seems to come from the fact that  nonholonomic dynamics predicts that ``...on
the average the skate does not slide off from the inclined plane''. For this reason, so the argument goes, one would be justified in   studying vakonomic mechanics as a viable competitor to nonholonomic mechanics. Accordingly, it has been suggested that the choice between vakonomic and nonholonomic mechanics in each concrete case can be answered only by experiment \cite{Arnold}. 
As a matter of fact, it is vakonomic mechanics that gives rise to unphysical behavior for a skate on an inclined plane: for  horizontal initial center-of-mass velocity and no initial rotation, in a neighborhood of $t = 0$ the skate goes up the inclined plane \cite{Zampieri}.

Rotational dynamics is full of surprises and rich in examples that violate our intuitive expectations, cases in point being the tippe top \cite{Cohen},  the spinning cylinder \cite{Jackson}, the asymmetric top \cite{Mecholsky} and the rattleback  \cite{Jones}. Intuition may fail even in  elementary rigid body dynamics \cite{Lemos3}. Although the motion predicted by nonholonomic mechanics for the coin's center of mass is counterintuitive, it is actually realized in Nature.

\begin{figure}[t!]
\epsfxsize=10cm
\begin{center}
\leavevmode
\epsffile{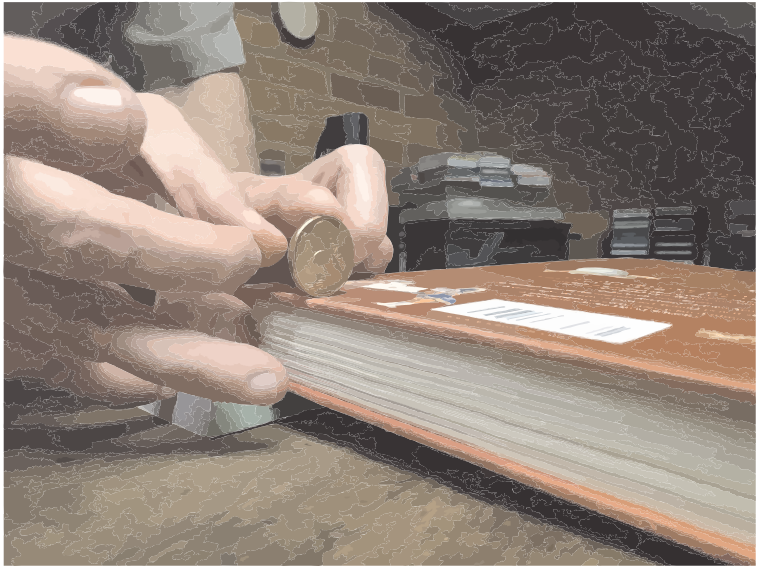}
\end{center}
\caption{Upright coin about to be flicked at the rim with a finger and set in motion with the initial conditions (\ref{initial-conditions}).}
\label{InitialCoin}
\end{figure}

\subsection{Rolling coin dynamics according to vakonomic mechanics}

From the Lagrangian (\ref{lagrangiana-moeda-rolante}) and the constraint functions 
(\ref{constraint-functions}) it follows that the vakonomic equations of  motion (\ref{vakonomic-equations-of-motion}) become
\begin{eqnarray}
 m{\ddot x} & = & - {\dot \lambda_1} ,\label{moeda-rolante-equacao-movimento-vakonomic-x}\\ 
 m{\ddot y}  & = &  mg\sin\alpha - {\dot \lambda_2} ,\label{moeda-rolante-equacao-movimento-vakonomic-y}\\
 \frac{mR^2}{2}\, {\ddot \phi} & = & \lambda_1 R{\dot \theta}\cos\theta - \lambda_2 R{\dot \theta}\sin\theta +  {\dot \lambda_1} R \sin\theta +  {\dot \lambda_2} R \cos\theta   , 
\label{moeda-rolante-equacao-movimento-vakonomic-phi} \\
 \frac{mR^2}{4}\, {\ddot \theta} & = &  -\lambda_1 R{\dot \phi}\cos\theta  + \lambda_2 R{\dot \phi}\sin\theta  . 
\label{moeda-rolante-equacao-movimento-vakonomic-theta}
\end{eqnarray}
These equations of motion differ in substantial aspects from the nonholonomic equations of motion. The most striking  discrepancy is between equations (\ref{moeda-rolante-equacao-movimento-vakonomic-theta}) and (\ref{moeda-rolante-equacao-movimento-c}). 

Can some   nonholonomic motion be brought about by equations (\ref{moeda-rolante-equacao-movimento-vakonomic-x}) to (\ref{moeda-rolante-equacao-movimento-vakonomic-theta})? We proceed to show that this question is  answered in the negative. 

\subsection{Complete inequivalence of  vakonomic and nonholonomic mechanics}

The following is the main result of this paper. 

\begin{theorem*}
Except for the trivial case $\theta = \mbox{constant}$, in which the constraints (\ref{constraints}) are actually  holonomic,   vakonomic mechanics is completely inequivalent to nonholonomic mechanics  for the rolling coin on an inclined plane.
\end{theorem*}

\begin{proof}
Our purpose is to prove that equations (\ref{moeda-rolante-equacao-movimento-vakonomic-x})-(\ref{moeda-rolante-equacao-movimento-vakonomic-theta}) do not possess solutions of the form (\ref{moeda-rolante-solucao-theta}), (\ref{moeda-rolante-solucao-para-phi}), (\ref{moeda-rolante-solucao-para-x}) and  (\ref{moeda-rolante-solucao-para-y}) no matter what the values of $\lambda_1(0)$ and $\lambda_2(0)$ may be.
We start by noting that  equation (\ref{moeda-rolante-equacao-movimento-vakonomic-x})   leads to
\begin{equation}
\label{inequivalence-lambda1}
 \lambda_1 (t) - \lambda_1 (0) = -m [ {\dot x}(t) - {\dot x}(0) ].
\end{equation}
Similarly, from  equation (\ref{moeda-rolante-equacao-movimento-vakonomic-y}) we get
\begin{equation}
\label{inequivalence-lambda2}
 \lambda_2 (t) - \lambda_2 (0) = -m [ {\dot y}(t) - {\dot y}(0) ] + mg t\sin\alpha .
\end{equation} 
Since ${\ddot \theta} =0$ from (\ref{moeda-rolante-solucao-theta}), equation (\ref{moeda-rolante-equacao-movimento-vakonomic-theta}) yields
\begin{equation}
\label{inequivalence-lambda1-and-2}
 \lambda_1 (t){\dot y}(t) =  \lambda_2 (t){\dot x}(t) ,
\end{equation}
where the constraints (\ref{constraints}) have been used. Insertion of (\ref{inequivalence-lambda1})  and (\ref{inequivalence-lambda2}) into (\ref{inequivalence-lambda1-and-2}) leads to
\begin{equation}
\label{inequivalence-nearly-final}
 \bigl[ \lambda_1 (0) + m{\dot x}(0)\bigr] \, {\dot y}(t) =  \bigl[\lambda_2(0) + m{\dot y}(0) + (mg \sin \alpha )t\bigr] \, {\dot x}(t).
\end{equation}
Suppose $\theta \neq \mbox{constant}$, so that $ \Omega \neq 0$. Consider the sequence $\{ t_n\} $ defined by $\Omega t_n + \theta_0 = 2n \pi$ where $n\in\mathbb{N}$. The sequence $\{ t_n\} $  is such that $t_n  \to \pm\infty$ as $n \to \infty$  according as $\Omega >0$ or $\Omega < 0$. Equations  (\ref{moeda-rolante-solucao-para-x}) and (\ref{moeda-rolante-solucao-para-y}) give
\begin{equation}
\label{x-ponto-y-ponto-tn}
{\dot x}(t_n) = \frac{2g \sin \alpha }{3\Omega}, \,\,\,\,\,\,\,\,\,\, {\dot y}(t_n) = \omega R.
\end{equation}
With these results, equation (\ref{inequivalence-nearly-final}) evaluated at $t=t_n$ reduces to
\begin{equation}
\label{inequivalence-nearly-final-tn}
 \bigl[\lambda_1 (0) + m{\dot x}(0)\bigr] \omega R =  \bigl[\lambda_2(0) + m{\dot y}(0) + (mg \sin \alpha )t_n\bigr] \frac{2g \sin \alpha }{3\Omega}.
\end{equation}
Since the left-hand side of this equation is a constant whereas its  right-hand side tends to plus  infinity as $n \to \infty$, equation (\ref{inequivalence-nearly-final})  
cannot be satisfied whatever the  values of $\lambda_1 (0)$ and $\lambda_2 (0)$ may be.
\end{proof}

\noindent {\bf Remark.} \hspace{.1cm} If $\alpha =0$,   the expression within brackets  on the right-hand side of (\ref{inequivalence-nearly-final}) no longer contains a term proportional to $t$. As a consequence, equation (\ref{inequivalence-nearly-final}) can be  satisfied by the choices $ \lambda_1 (0) = - m{\dot x}(0)$ and $\lambda_2(0) =- m{\dot y}(0)$,  in agreement with the results in \cite{Cortes1} and  \cite{Fernandez} 
  regarding the rolling coin on a horizontal plane.

Finally, it is to be noted that, for a ball rolling without slipping on a rotating table, Lewis and  Murray \cite{Lewis}  showed that for  {\it some particular initial conditions} the nonholonomic motion is not  vakonomic. Here, for the rolling coin on an inclined plane,  we have proved  that for {\it all nontrivial initial conditions} the nonholonomic motion is inaccessible  from  vakonomic  mechanics.

\section{Conclusions}\label{Conclusions} 

The proposal of vakonomic mechanics as a viable model to describe the dynamics of nonholonomic systems is fraught with   predicaments. 

From a theoretical point of view, the presence of both $\lambda_l$ and ${\dot \lambda}_l$ in the equations of motion of   vakonomic mechanics means that in order to determine the motion uniquely one has to specify $\lambda_l(0)$, the values  of the  Lagrange multipliers at the initial instant $t=0$. This implies that the well-tested Newtonian  mechanics would have to be replaced by a new mechanics containing hidden variables, whose physical meaning is mysterious. 
Moreover, this entails, given that the Lagrange multipliers are directly connected to the constraint forces, that the initial accelerations would have to be prescribed together with the initial positions and velocities, which by all available evidence is not physically sound.  Hidden-variable quantum theories are not successful as alternatives to quantum mechanics, but the lattter's probabilistic nature justifies the attempts to build them. Since macroscopic physics is deterministic, there does not appear  to be a corresponding motivation to search for  hidden variables in  
classical mechanics. Thus, it seems fair to say that the physical significance of vakonomic mechanics, if any, remains opaque \cite{Kharlamov,Flannery2,Flannery3}.

As to the actual behavior of nonholonomic systems, 
experimental studies \cite{Lewis,Kai} favor (\ref{nonholonomic-equations-of-motion}) as the correct equations of motion. The qualitave home experiment described in the present paper corroborates these findings.

If, as held by advocates  of vakonomic mechanics \cite{Arnold},    experiment alone can decide  between vakonomic and nonholonomic mechanics in each concrete case, experiment has been  speaking loudly for nonholonomic mechanics and against vakonomic mechanics.

\end{document}